\documentclass[11pt]{article}

\usepackage{graphicx}
\usepackage{float}
\usepackage[ruled,vlined]{algorithm2e}
\usepackage{caption}
\usepackage{amsmath}
\usepackage{amssymb}
\usepackage{amsmath}
\usepackage{amsthm}
\usepackage{titlesec}
\usepackage{indentfirst}

\newtheorem{corollary}{Corollary}
\newtheorem{definition}{Definition}

\newtheorem{lemma}{Lemma}

\newtheorem{remark}{Remark}
\newtheorem{theorem}{Theorem}
\begin{document}
\title{{\bf Distributed Grover's algorithm}}
\author{Daowen Qiu$^{1,3,*}$, Le Luo$^{2,3}$, Ligang Xiao$^{1,3}$\\
\small{$^{1}$School of Computer Science and Engineering, Sun Yat-sen University, Guangzhou,} {\small China}\\
\small{$^{2}$School of Physics and Astronomy, Sun Yat-sen University, {\rm  Zhuhai}, }{\small China}\\
\small{$^{3}$QUDOOR Technologies Inc.,  {\rm Guangzhou},} {\small China}\\
{\small $^*$Corresponding author.  E-mail address: issqdw@mail.sysu.edu.cn}}

\date{  }

\maketitle
\begin{center}
\begin{minipage}{130mm}
\begin{center}{\bf Abstract}\end{center}
{\small
\vskip 1mm \hskip 4.7mm 
 Let Boolean function $f:\{0,1\}^n\longrightarrow \{0,1\}$ where $|\{x\in\{0,1\}^n| f(x)=1\}|=a\geq 1$. To search for an $x\in\{0,1\}^n$ with $f(x)=1$, by Grover's algorithm we can get the objective with query times $\lfloor \frac{\pi}{4}\sqrt{\frac{2^n}{a}} \rfloor$. In this paper, we propose a distributed Grover's algorithm for computing $f$ with lower query times and smaller number of input bits. More exactly, for any $k$ with $n>k\geq 1$, we can decompose $f$ into $2^k$ subfunctions, each which has $n-k$ input bits, and then the objective can be found out by computing these subfunctions with query times at most $\sum_{i=1}^{r_i}  \lfloor \frac{\pi}{4}\sqrt{\frac{2^{n-k}}{b_i}} \rfloor+\lceil\sqrt{2^{n-k}}\rceil+2t_a+1$ for some $1\leq b_i\leq a$ and $r_i\leq 2t_a+1$, where $t_a=\lceil 2\pi\sqrt{a}+11\rceil$. In particular, if $a=1$, then our distributed Grover's algorithm only needs $\lfloor \frac{\pi}{4}\sqrt{2^{n-k}} \rfloor$ queries, versus $\lfloor \frac{\pi}{4}\sqrt{2^{n}} \rfloor$ queries of Grover's algorithm. 
 Finally, we propose an efficient algorithm of constructing quantum circuits for  realizing the oracle corresponding to any Boolean function with conjunctive normal form (CNF).

}

\vskip 2mm

\end{minipage}
\end{center}
\vskip 10mm

\section{ Introduction}

Quantum computers  were first considered by Benioff \cite{Ben80},
and then suggested by Feynman \cite{Fey82}  in 1982. By
formalizing Benioff and Feynman's ideas,  Deutsch \cite{Deu85} in 1985
proved the existence of universal {\it quantum
Turing machines} (QTMs) and proposed the quantum Church-Turing Thesis. Subsequently, Deutsch \cite{Deu89}
considered quantum network models. In 1993, Yao \cite{Yao93}
elaborated on  the simulation of QTMs by quantum circuits (recently simulation has been further improved \cite{MW18}).
Universal QTMs simulating other QTMs with polynomial time was proved
by Bernstein and Vazirani \cite{BV97}.

The Deutsch-Jozsa algorithm was proposed by Deutsch and Jozsa in 1992 \cite{DJ92} and improved by Cleve,  Ekert, Macchiavello, and  Mosca in 1998 \cite{CEMM98}. For determining whether a function $f:\{0,1\}^n\longrightarrow \{0,1\}$ is constant or balanced, the  Deutsch-Jozsa algorithm \cite{DJ92,CEMM98} solves exactly the Deutsch-Jozsa problem   with exact quantum 1-query, but the classical  algorithm requires $2^{n-1}+1$ queries to compute it deterministically.  It presents a basic procedure of quantum algorithms, and in a way provides inspiration for Simon's algorithm, Shor's algorithm, and Grover's algorithm \cite{KLM07,NC00}.

 Deutsch-Jozsa problem and Simon problem have been further studied.  In fact,  it was proved that all  symmetric partial Boolean functions with exact quantum 1-query complexity can be computed by the  Deutsch-Jozsa algorithm \cite{QZ20}, and  further generalization of Deutsch-Jozsa problem was studied in \cite{QZ18}. An optimal separation between exact quantum query complexity and classical deterministic query complexity for Simon problem was obtained in \cite{CQ18}.

Grover's algorithm can find out one target element in an unordered database if the sum of all target elements is known.  More formally,  let a Boolean function $f:\{0,1\}^n\longrightarrow \{0,1\}$ where $|\{x\in\{0,1\}^n| f(x)=1\}|=a\geq 1$. To search for an $x\in\{0,1\}^n$ with $f(x)=1$, by Grover's algorithm we can get the objective with  $\lfloor \frac{\pi}{4}\sqrt{\frac{2^n}{a}} \rfloor$ queries, and the success probability is close to $1$. However, any classical algorithm to solve it needs $\Omega(2^n)$ queries. Grover's algorithm with zero theoretical failure rate was considered by Long \cite{Long01}.  The problem of operator coherence in Grover's algorithm was considered in \cite{PQ19}.

After Grover's algorithm, the algorithms of quantum amplitude amplification and estimation \cite{BHMT02} were proposed and developed. The algorithm of quantum amplitude amplification is a generalization of Grover's algorithm. We describe the algorithm of quantum amplitude estimation \cite{BHMT02} roughly. Given a Boolean function $f:\{0,1\}^n\longrightarrow \{0,1\}$, and a quantum algorithm ${\cal A}$ acting on $|0^n\rangle$, then we hope to get the quantity of information for $f(x)=1$ from ${\cal A}|0^n\rangle$. Actually, the algorithm of quantum amplitude estimation (more exactly, quantum counting) \cite{BHMT02} can answer this problem by making $\lceil\sqrt{2^{n}}\rceil$ queries on $f$ with high success probability.
Therefore, if the sum of all target elements in an unordered database  is not known, then  we can employ the algorithm of quantum counting \cite{BHMT02} to estimate it  with high success probability (see Corollary \ref{estimate}).

 The exponential speed-up of Shor's quantum algorithm for factoring
integers in polynomial time \cite{Sho97} and afterwards Grover's
algorithm of searching in database of size $N$ with only
$O(\sqrt{N})$ accesses \cite{Gro96} have already shown great advantages of quantum computing over classical computing, but nowadays  it is still difficult to build large-scale universal quantum computers due to noise and depth of quantum circuits. So, in the NISQ (Noisy Intermediate-scale Quantum) era, developing new quantum algorithms and models with better physical realizability is an intriguing and promising research field, and distributed quantum computing is such a feasible and useful subject.

Distributed quantum computing has been studied from different methods and points of view (for example,  \cite{ACR21, BBG13, LQ17, TXQ22, YL04}). As we are aware, there are three methods for distributed quantum computing in general. One way is directly to divide the quantum circuit for computing a problem into multiple quantum sub-circuits, but quantum communications (such as quantum teleportation) are needed among quantum sub-circuits, for example, a distributed Shor's algorithm in \cite{YL04}. However, the price of this method is more teleportations to be paid; the second method is to get multiple local solutions by using similar quantum algorithms to the original one and then conclude the final solution of problem, for example, distributed quantum phase estimation in \cite{LQ17}; the third method proposed recently is decomposing a Boolean function to be computed into multiple subfunctions \cite{ACR21}, and then computing these (all or partial) subfunctions to get the solution of original problem (for example \cite{TXQ22}). In fact, we have proposed a distributed algorithm for solving Simon's problem, where there exist actually entanglement among those oracles for querying subfunctions.

 In this paper, we use the third method to design a distributed Grover's algorithm, but these oracles for querying all subfunctions can work separately in our algorithm, that is, without entnaglement between these oracles. In addition, we propose an efficient algorithm with time complexity $O(m  \log m)$ for realizing the oracle (i.e., unitary operator $Z_f$), where  $Z_f|x\rangle=(-1)^{f(x)}|x\rangle$, and $f$ is a Boolean fucntion  with  conjunctive normal form (CNF) having $m$ clauses.

The remainder of the paper is organized as follows. In Section 2, we recall the  Grover's algorithm and  the algorithms of quantum amplitude  estimation and quantum counting that we will use in this paper. Then, in Section 3, we give two distributed Grover's algorithms, where one is serial and another is parallel. In our distributed algorithms, many oracles are required, so in Section 4, we propose an efficient algorithm for realizing the oracle corresponding to any Boolea function with  conjunctive normal form (CNF). Finally, in Section 5 we summarize the results obtained and mention related problems for further study.

\section{ Preliminaries}

In the interest of readability,  in this section, we briefly review the Grover's algorithm \cite{Gro96} and the algorithms of quantum amplitude  estimation and quantum counting \cite{BHMT02}.  
First we recall Grover's algorithm.

\subsection{Grover's algorithm}

 Grover's algorithm aims to search for a goal (solution) to a very
wide range of problems. We formally describe this problem  as follows.  Define a
function $f:\{0,1\}^n \rightarrow\{0,1\}$ and we assume that $f(x)=1$ for some $x\in\{0,1\}^n$ (here we consider $|\{x\in\{0,1\}^n| f(x)=1\}|=a\geq 1$). The goal is to find out one $x$ such that $f(x)=1$. So, the search problem is as follows.

Input: A function $f:\{0,1\}^n \rightarrow\{0,1\}$ with $|\{x\in\{0,1\}^n| f(x)=1\}|=a\geq 1$

Problem: Find an input $x\in\{0,1\}^n$ such that $f(x) = 1$.

Grover's algorithm performs the search quadratically faster than any classical algorithm.  In fact, any classical algorithm
finding a solution with probability at least $\frac{2}{3}$ must make
$\Omega(2^n)$ queries in the worst case, but Grover's  algorithm takes only
$O(\sqrt{2^n})$ queries, more exactly $\lfloor \frac{\pi}{4}\sqrt{\frac{2^n}{a}} \rfloor$ queries.

Although this is not the exponential quantum advantage achieved
as Shor's algorithm for factoring, the wide applicability of searching
problems makes Grover's algorithm interesting and important.
We describe the Grover's algorithm in the following Figure 1, where $N=2^n$.

\begin{figure}[H]
    \centering
    \includegraphics[scale=0.65]{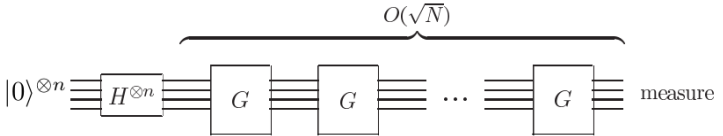}
    \caption{Grover's algorithm }
\end{figure}

\begin{algorithm}[H]
	\caption{Grover's Algorithm}\label{Grover's Algorithm}
	\LinesNumbered 
	\KwIn{A function $f:\{0,1\}^n \rightarrow\{0,1\}$ with $|\{x\in\{0,1\}^n| f(x)=1\}|=a\geq 1$.}
	\KwOut{the string $x\in\{0,1\}^n$ such that $f(x)=1$.}
	$H^{\otimes n}|0\rangle^{\otimes n}\rightarrow\frac{1}{\sqrt{2^n}}\sum\limits_{x\in\{0 , 1\}^n}|x\rangle$.
	
	 $G$ is performed with $\lfloor \frac{\pi}{4}\sqrt{\frac{2^n}{a}} \rfloor$ times, which $G=-H^{\otimes n}Z_0H^{\otimes n}Z_f$, $Z_f|x\rangle=(-1)^{f(x)}|x\rangle$, $Z_0|x\rangle=\left\{
\begin{array}{rcl}
-|x\rangle, & &x=0^n ;\\
|x\rangle, & & x\neq 0^n.
\end{array} \right.$
	
	Measure the resulting state.
	
\end{algorithm}

Next, we recall the  analyse of  Grover's algorithm. Denote
\begin{align}
&A=\{x\in\{0 , 1\}^n | f(x)=1\} 
,\\
&B=\{x\in\{0 , 1\}^n | f(x)=0\} 
,\\
&|A\rangle=\frac{1}{\sqrt{a}}\sum\limits_{x\in A}|x\rangle,\\
&|B\rangle=\frac{1}{\sqrt{b}}\sum\limits_{x\in B}|x\rangle.
\end{align}

Then after step 1 of  algorithm, the quantum state is
\begin{equation}
H^{\otimes n}|0\rangle^{\otimes n}=\sqrt{\frac{a}{N}}|A\rangle+\sqrt{\frac{b}{N}}|B\rangle\triangleq|h\rangle,
\end{equation}
where $a=|A|$, $b=|B|$, $N=2^n$.

It is easy to know that $Z_0=I-2|0^n\rangle\langle 0^n|$, so
\begin{align}
H^{\otimes n}Z_0H^{\otimes n}&=H^{\otimes n}(I-2|0^n\rangle\langle 0^n|)H^{\otimes n}\\
&=I-2|h\rangle\langle h|.
\end{align}
 Let $\sin{\theta}=\sqrt{\frac{a}{N}}$, $\cos{\theta}=\sqrt{\frac{b}{N}}$.
Then
we have
\begin{equation}
G^k|h\rangle=\cos{((2k+1)\theta)}|B\rangle+\sin{((2k+1)\theta)}|A\rangle.
\end{equation}
Our goal is to make $\sin{((2k+1)\theta)}\approx 1$,  that is $(2k+1)\theta\approx\frac{\pi}{2}$, so $k\approx\frac{\pi}{4\theta}-\frac{1}{2}$.


Since $\theta=\sin^{-1}{\sqrt{\dfrac{a}{N}}}\approx    \sqrt{\dfrac{a}{N}} $,  we obtain $k=\lfloor \frac{\pi}{4}\sqrt{\frac{2^n}{a}} \rfloor$, and  the probability of success is:
\begin{equation}
\sin^2{((2\lfloor \pi\sqrt{N}/4\sqrt{a} \rfloor+1)\sin^{-1}(\sqrt{a}/\sqrt{N}))},
\end{equation}
where the probability is close to $1$.

\subsection{ Algorithm of quantum amplitude estimation}

In this subsection  we introduce the algorithm of quantum amplitude estimation \cite{BHMT02}.

Let $f:\{0,1\}^n\rightarrow\{0,1\}$. Suppose 
we have a quantum algorithm $\mathcal{A}$ without measurements  such that $\mathcal{A}|0^n\rangle=|\Psi\rangle$. Further, there are:
\begin{equation}|\Psi\rangle=|\Psi_1\rangle+|\Psi_0\rangle, \end{equation} where
\begin{align}
&|\Psi_1\rangle=\sum\limits_{x:f(x)=1}\alpha_x|x\rangle,\\
&|\Psi_0\rangle=\sum\limits_{x:f(x)=0}\alpha_x|x\rangle.
\end{align}
Denote

 \begin{equation}a_g=\langle\Psi_1|\Psi_1\rangle , 0<a_g<1. \end{equation}
Let
\begin{align*}
&Q=\mathcal{A}U_0^{\perp}\mathcal{A}^{-1}U_f,\\
&U_0^{\perp}=2|0^n\rangle \langle0^n|-I,\\
&U_f|x\rangle=(-1)^{f(x)}|x\rangle.
\end{align*}
Define
\begin{equation}
QFT_{2^m}|j\rangle=\frac{1}{\sqrt{2^m}}\sum\limits_{k=0}^{2^m-1} e^{2\pi ijk/2^m}|k\rangle,
\end{equation}
where $0\leq j<2^m-1$, and it follows that
\begin{equation}
QFT_{2^m}^{\dagger}|j\rangle=\frac{1}{\sqrt{2^m}}\sum\limits_{k=0}^{2^m-1} e^{-2\pi ijk/2^m}|k\rangle,
\end{equation}
where $0\leq j<2^m-1$.
Suppose   unitary transformation $Q$ acts on $n$ qubits, $m$ is a positive integer,  $\Lambda_{2^m}(Q)$ represents  a unitary transformation acting on $m+n$ qubits and is defined as:
\begin{equation}
\Lambda_{2^m}(Q)|j\rangle|\psi\rangle=|j\rangle(Q^j|\psi\rangle),
\end{equation}
where $j\in\{0, \cdots, 2^m-1\}$, $|\psi\rangle$ is the  state corresponding to the $n$ qubit.

When $a_g$ is unknown, we can use quantum amplitude estimation to get an approximation of $a_g$.

\begin{figure}[H]
    \centering
    \includegraphics[scale=0.6]{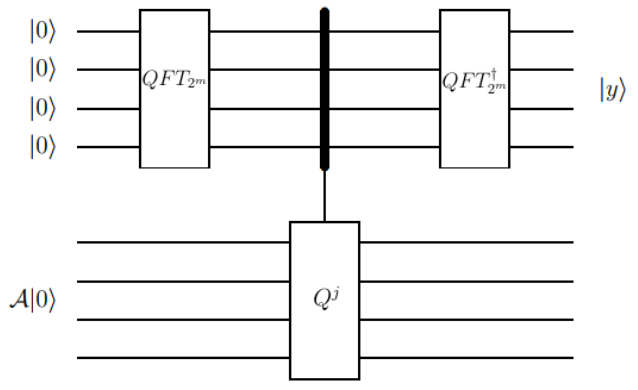}
    \caption{Algorithm of quantum amplitude estimation }
\end{figure}

\begin{algorithm}[H] \label{AE}
	\caption{Algorithm of quantum amplitude estimation ${\rm Est\_ Amp}(\mathcal{A}, f, 2^m)$}
	\LinesNumbered 
	\KwIn{$\mathcal{A}$, $f$, $2^m$.}
	\KwOut{$\widetilde{a}=\sin^2(\pi\frac{y}{2^m})$.}
	 Initialize two registers  to the state $|0^m\rangle\mathcal{A}|0^n\rangle$.
	
	 Apply $QFT_{2^m}$ to the first register, $(QFT_{2^m}\otimes I)|0^m\rangle\mathcal{A}|0^n\rangle$.
	
	 Apply $\Lambda_{2^m}(Q)$  to the state after step 2.
	
	 Apply $(QFT_{2^m}^{\dagger}\otimes I)$  to the state after step 3.
	
	 Measure the first register and denote the outcome $|y\rangle$.
	
	 Output $\widetilde{a}=\sin^2(\pi\frac{y}{2^m})$.
	
\end{algorithm}

To illustrate the above algorithm, we recall the following theorem \cite{BHMT02}.

\begin{theorem} \cite{BHMT02} \label{AEP}
For any positive integer $k$, the algorithm ${\rm Est\_Amp}(\mathcal{A}, f, 2^m)$ outputs $\widetilde{a}$ $(0\leq \widetilde{a}\leq 1)$  such that
\begin{align}
|\widetilde{a}-a_g|\leq2\pi k\frac{\sqrt{a_g(1-a_g)}}{2^m}+k^2\frac{\pi^2}{2^{2m}},
\end{align}
the probability that this holds is:

$(1)$ When $k=1$, $p\geq\frac{8}{\pi^2}$;

$(2)$ When $k\geq 2$, $p\geq1-\frac{1}{2(k-1)}$.

If $a_g = 0$ then $\widetilde{a}= 0$ with certainty, and if $a_g = 1$  then $\widetilde{a} = 1$ with certainty.
\end{theorem}

We recall  the first stage of quantum phase estimation procedure (Figure 3) that can implement $\Lambda_{2^m}(Q)$, where we replace $U$ by $Q$ in Figure 3. Each $Q$ contains one $U_f$, so the total number of queries on $U_f$ is $\sum_{i=0}^{m-1}2^i=2^m-1$.

\begin{figure}[H]
    \centering
    \includegraphics[scale=0.6]{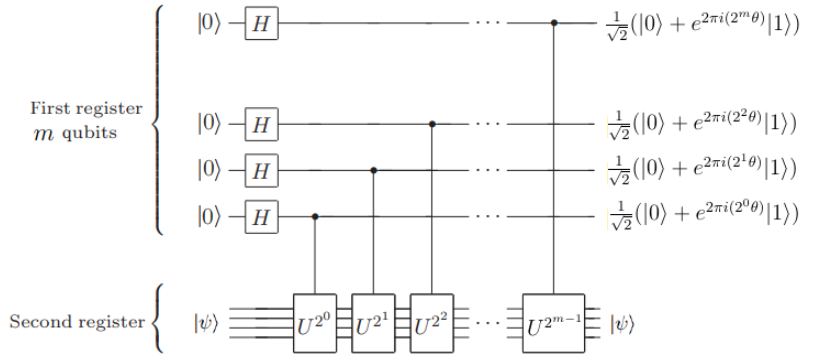}
    \caption{The first stage of quantum phase estimation procedure}
\end{figure}

A straightforward application of this algorithm is to approximately count the number of solutions $t$ to $f(x) = 1$ where $f : \{0,1\}^n\rightarrow\{0,1\}$. We then have $a_g=\langle\Psi_1|\Psi_1\rangle=\frac{t}{2^n}$. Actually, by means of the above theorem, Brassard et al \cite{BHMT02} further presented an algorithm of quantum counting  and obtained the following theorem. Here let $\mathcal{A}$ be  Hadamard transformation $H^{\otimes n}$.


\begin{algorithm}[H]
	\caption{Quantum counting algorithm ${\rm \textbf{Count}}(\mathcal{A}, f, 2^m)$}
	\LinesNumbered 
	\KwIn{$\mathcal{A}$, $f$, $2^m$.}
	\KwOut{$t'=2^n\times \sin^2(\pi\frac{y}{2^m})$.}
    Initialize two registers  to the state $|0^m\rangle\mathcal{A}|0^n\rangle$.
	
	 Apply ${\rm{QFT}}_{2^m}$ to the first register, $({\rm{QFT}}_{2^m}\otimes I)|0^m\rangle\mathcal{A}|0^n\rangle$.
	
	 Apply $\Lambda_{2^m}(Q)$  to the state after step 2.
	
	 Apply $({\rm{QFT}}_{2^m}^{\dagger}\otimes I)$  to the state after step 3.
	
	 Measure the first register and denote the outcome $|y\rangle$.
	
	 Output $t'=2^n\times \sin^2(\pi\frac{y}{2^m})$.
\end{algorithm}
By Theorem \ref{AEP},  we obtain the following theorem  by taking  $\mathcal{A}$ to be the Hadamard transformation $H^{\otimes n}$.

\begin{theorem}
For any positive integer $k$, and any Boolean function $f : \{0,1\}^n\rightarrow\{0,1\}$,  the algorithm ${\rm \textbf{Count}}(H^{\otimes n}, f, 2^m)$ outputs an estimate $t'$ to $t=|f^{-1}(1)|$ such that
\begin{align}
|t'-t|\leq2\pi k\frac{\sqrt{t(2^n-t)}}{2^m}+k^2\frac{\pi^2 2^n}{2^{2m}},
\end{align}
the probability that this holds is:

$(1)$ When $k=1$ , $p\geq\frac{8}{\pi^2}$;

$(2)$ When $k\geq 2$ , $p\geq1-\frac{1}{2(k-1)}$.

If $t = 0$ then $t'= 0$ with certainty, and if $t = 2^n$  then $t'= 2^n$ with certainty.
\end{theorem}

In particular, if we want to estimate $t$ within a few  deviations, we can apply algorithm ${\rm \textbf{Count}}(\mathcal{A}, f, 2^m)$ with $2^m=\lceil\sqrt{2^n}\rceil$. This is the following corollary  further presented in \cite{BHMT02}.

\begin{corollary} \cite{BHMT02} \label{estimate}
Given a Boolean function $f : \{0,1\}^n\rightarrow\{0,1\}$ with $t=|\{x\in\{0,1\}^n| f(x)=1\}|\geq 1$, there is an algorithm of quantum counting ${\rm \textbf{Count}}(H^{\otimes n}, f, \lceil\sqrt{2^n}\rceil)$ requiring exactly  $\lceil\sqrt{2^n}\rceil$ queries of  $f$ and outputting
integer number $t'$ such that
\begin{align}
|t'-t|\leq2\pi \sqrt{\frac{t(2^n-t)}{2^n}}+11
\end{align}
with probability at least $\frac{8}{\pi^2}$.  In particular, if $t = 0$ then $t'= 0$ is determined with certainty, and if $t = 2^n$  then $t'= 2^n$ is determined  with certainty.

\end{corollary}




\section{Distributed Grover's algorithm}

Let Boolean function $f:\{0,1\}^n \rightarrow\{0,1\}$, and suppose $|\{x\in\{0,1\}^n| f(x)=1\}|=a\geq 1$. For any $n>k\geq 1$, we divide $f$ into $2^k$ subfunctions as follows.

For any $i\in [0,2^k-1]$, we identify $y_i\in\{0,1\}^k$ with the binary representation of $i$, and then define Boolean function $f_i:\{0,1\}^{n-k}\rightarrow \{0,1\}$ as follows: For any $x\in\{0,1\}^{n-k}$,

\begin{equation} \label{DF}
f_i(x)=f(xy_i).
\end{equation}

If we can find out some $x\in\{0,1\}^{n-k}$ such that $f_i(x)=1$, then the solution has been discovered. However, if  function $f_i(x)=0$ for  any $x\in\{0,1\}^{n-k}$, then we cannot get any useful solution from it.  We employ the
 algorithm of quantum counting \cite{BHMT02} to determine the number of ``good" elements for each subfunction, where a ``good" element means an input string mapping to function  value $1$. If the subfunction is constant to $0$, then the algorithm of quantum counting can determine this subfunction is constant to $0$ exactly without error; if the subfunction is not constant to $0$, then  the  algorithm of quantum counting can determine the number of ``good" elements of subfunction with success probability at least $\frac{8}{\pi^2}$.

After that, for a subfunction having function value $1$, we use Grover's algorithm to find out a solution that is the goal of the original problem as well.

So, before we perform the Grover's algorithm for a subfunction, we use the algorithm of quantum counting \cite{BHMT02} to determine the number of ``good" elements of subfunctions.

First we present a lemma. Here we denote $|\{x\in\{0,1\}^n| f(x)=1\}|=a\geq 1$ as above.

\begin{lemma} \label{Count}
 Let Boolean function $f:\{0,1\}^n \rightarrow\{0,1\}$. Given $n>k\geq 1$,
for any integer $i\in [0,2^k-1]$, denote $a_i=|\{x\in\{0,1\}^{n-k}|f_i(x)=1\}|$ and $t_a=\lceil 2\pi\sqrt{a}+11\rceil$. Then the algorithm of quantum counting ${\rm \textbf{Count}}(H^{\otimes (n-k)}, f_i, \lceil\sqrt{2^{n-k}}\rceil)$
 can outputs an integer number $a_i'$ such that  $a_i\in \{a_i'-t_a, a_i'-t_a+1,\ldots, a_i'+t_a\}$  with $\lceil\sqrt{2^{n-k}}\rceil$ queries to $f_i$, and the success probability is at least $\frac{8}{\pi^2}$. In particular, if $a_i = 0$ then $a_i'= 0$ is determined with certainty, and if $a_i = 2^{n-k}$  then $a_i'= 2^{n-k}$ is determined  with certainty.

\end{lemma}

\begin{proof}
It follows from Corollary \ref{estimate} by taking $f_i$ there. More specifically, for inputting function $f_i$, by virtue of the algorithm of quantum counting ${\rm \textbf{Count}}(H^{\otimes (n-k)}, f_i, \lceil\sqrt{2^{n-k}}\rceil)$, an integer number $a_i'$ is outputted and satisfies
\begin{align}
|a_i'-a_i|&\leq 2\pi \sqrt{\frac{a_i(2^{n-k}-a_i)}{2^{n-k}}}+11\\
&\leq 2\pi \sqrt{a}+11\\
&=t_a,
\end{align}
where $a_i\leq a$ is used. Therefore, we have $a_i\in \{a_i'-t_a, a_i'-t_a+1,\ldots, a_i'+t_a\}$.
\end{proof}

\begin{remark}
From Lemma \ref{Count} we know that $a_i$ has at most $2t_a+1$ possible values.
\end{remark}

After $f_i$ is determined not constant to zero, we use Grover's algorithm to search for a solution $x\in\{0,1\}^{n-k}$ such that $f_i(xy_i)=1$. Then $xy_i$ is exactly the solution of the original problem.

For subfunctions $f_i$ $(i=0,1,2,\ldots,2^{k}-1)$, we can use $2^{k}$ machines in  parallel to deal with $f_i$ independently, or we use one machine only to do it in sequence. Next we analyze their query complexity respectively.

If we use $2^{k}$ machines in  parallel to compute, then once $f_i$ is determined as non-constant to zero by the algorithm of quantum counting, we further use Grover's algorithm to search for a solution. So, in this case, we need to query at most $\sum_{i=1}^{2t_a+1}  \lfloor \frac{\pi}{4}\sqrt{\frac{2^{n-k}}{b_i}} \rfloor+\lceil\sqrt{2^{n-k}}\rceil+2t_a+1$ times for some $1\leq b_i\leq a$ ($i=1,2,\ldots,2t_a+1$), where we use $\lceil\sqrt{2^{n-k}}\rceil$  queries in the algorithm of quantum counting, and $2t_a+1$ means the times to check the results for using Grover's algorithm regarding each possible number of goals $b_i$ of $f_i$.

If we use only one machine to do it with serial method, and suppose $f_i$ has no goals  (i.e., $a_i=0$) for $0\leq i\leq T\leq 2^{k}-2$, but $f_T$ does, then we need to query $\sum_{i=1}^{2t_a+1}  \lfloor \frac{\pi}{4}\sqrt{\frac{2^{n-k}}{b_i}} \rfloor+T\lceil\sqrt{2^{n-k}}\rceil+2t_a+1$ times at most (of course, if we are failure to get a solution from $f_T$, we can continue to search from $f_{T+1}$, but we do not consider this case here).


Next we give a distributed Grover's algorithm in serial method.

\begin{algorithm}[H] \label{DGAS}
	\caption{Distributed Grover's Algorithm (I)}
	\LinesNumbered 
	\KwIn{A function $f:\{0,1\}^n \rightarrow\{0,1\}$ with $f(x) = 1$  for some $x\in\{0,1\}^n$.}
	\KwOut{the string $x\in\{0,1\}^n$ such that $f(x)=1$.}
	
Given $k\in [1,n)$, decompose $f$ into  $f_i$ $(i=0,1,2,\ldots,2^{k}-1)$ as Equation (\ref{DF}).

Take $i=0$.

Use the algorithm of quantum counting ${\rm \textbf{Count}}(H^{\otimes (n-k)}, f_i, \lceil\sqrt{2^{n-k}}\rceil)$ to determine a range of $a_i=|\{x\in\{0,1\}^{n-k}|f_i(x)=1\}|$  with at most $2t_a+1$ possible values, say $\{b_1,b_2,\ldots,b_{2t_a+1}\}$ where $b_j\neq 0$ otherwise it is removed. If $a_i=0$ and $i<2^k-2$, return to Step 2 by taking $i=i+1$, otherwise go to next step by taking $j=1$.

$H^{\otimes (n-k)}|0\rangle^{\otimes (n-k)}\rightarrow\frac{1}{\sqrt{2^{n-k}}}\sum\limits_{x\in\{0 , 1\}^{n-k}}|x\rangle$.
	
	 $G_i$ is performed with $\lfloor \frac{\pi}{4}\sqrt{\frac{2^{n-k}}{b_i}} \rfloor$ times, where $G_i=-H^{\otimes (n-k)}Z_0H^{\otimes (n-k)}Z_{f_i}$, $Z_{f_i}|x\rangle=(-1)^{f_i(x)}|x\rangle$, $Z_0|x\rangle=\left\{
\begin{array}{rcl}
-|x\rangle, & &x=0^{n-k} ;\\
|x\rangle, & & x\neq 0^{n-k}.
\end{array} \right.$
	
	Measure the resulting state, and check whether it is a solution of $f_i$. If it is, then end the algorithm, otherwise take $j=j+1$ and return to Step 4.
	
\end{algorithm}

\begin{remark}

In above algorithm, we have not considered the case of being failure to get a solution from $f_T$ and then continuing to search from next subfunction.

\end{remark}

We can describe the above algorithm with the following theorem.

\begin{theorem} Given $n>k\geq 1$, let Boolean function $f:\{0,1\}^n \rightarrow\{0,1\}$ with $|x\in\{0,1\}^{n}|f(x)=1\}|=a\geq 1$. Then Algorithm  \ref{DGAS} can find out the string $x\in\{0,1\}^n$ such that $f(x)=1$ with query times at most
$\lfloor \frac{\pi}{4}\sqrt{2^{n-k}} \rfloor+(2^k-1)(4\sqrt{2^{n-k}}-1)$  and the success probability is at least
\begin{equation}
\frac{8}{\pi^2}P_{G_i},
\end{equation}
where $P_{G_i}=\sin^2{((2\lfloor\frac{ \pi\sqrt{2^{(n-k)}}}{4\sqrt{a_i}} \rfloor+1)\sin^{-1}(\frac{a_i}{\sqrt{2^{(n-k)}}}))}$ is the success probability of Grover's algorithm for computing the first $f_i$ with $f_i$ being not constant to zero.

\end{theorem}

\begin{proof}

Suppose $f_T$ is not a constant function, but $f_j\equiv 0$ for $j=0,1,2,\ldots,T-1$. Then the algorithm of quantum counting can exactly determine $f_j$ is constant to zero for $j=0,1,2,\ldots,T-1$. For $f_i$, the algorithm of quantum counting outputs a range of $a_T$ with success probability at least $\frac{8}{\pi^2}$, say $\{b_1,b_2,\ldots,b_{2t_a+1}\}$ where $b_i\neq 0$ otherwise it is removed (recall $a_T=|\{x\in\{0,1\}^{n-k}|f_T(x)=1\}|$).

After that, we use Grover's algorithm to compute $f_T$ for each possible number $b_j$ of goals and get the solution with success probability $P_{G_T}$, where $$P_{G_T}=\sin^2{((2\lfloor\frac{ \pi\sqrt{2^{(n-k)}}}{4\sqrt{a_T}} \rfloor+1)\sin^{-1}(\frac{a_T}{\sqrt{2^{(n-k)}}}))}.$$
So, in this case, the number of queries is  $\sum_{i=1}^{2t_a+1}  \lfloor \frac{\pi}{4}\sqrt{\frac{2^{n-k}}{b_i}} \rfloor+T\lceil\sqrt{2^{n-k}}\rceil+2t_a+1$, where we use $T\lceil\sqrt{2^{n-k}}\rceil$  queries in the algorithm of quantum counting, and $2t_a+1$ means the times to check the results for using Grover's algorithm regarding each possible number of goals $b_i$ of $f_T$. It is clear that

\begin{align}
&\sum_{i=1}^{2t_a+1}  \lfloor \frac{\pi}{4}\sqrt{\frac{2^{n-k}}{b_i}} \rfloor+T\lceil\sqrt{2^{n-k}}\rceil+2t_a+1 \\
&\leq (2t_a+1)  \lfloor \frac{\pi}{4}\sqrt{2^{n-k}} \rfloor+2^k\lceil\sqrt{2^{n-k}}\rceil+2t_a+1.
\end{align}

Therefore, the success probability is at least $\frac{8}{\pi^2}P_{G_i}.$ Since both $\frac{8}{\pi^2}$ and $P_{G_i}$ approximate $1$, the success probability is still high.

\end{proof}
Next we give a distributed Grover's algorithm in parallel.

\begin{algorithm}[H] \label{DGAP}
	\caption{Distributed Grover's Algorithm (II)}
	\LinesNumbered 
	\KwIn{A function $f:\{0,1\}^n \rightarrow\{0,1\}$ with $f(x) = 1$  for some $x\in\{0,1\}^n$.}
	\KwOut{the string $x\in\{0,1\}^n$ such that $f(x)=1$.}
	
Given $k\in [1,n)$, decompose $f$ into  $f_i$ $(i=0,1,2,\ldots,2^{k}-1)$ as Equation (\ref{DF}).

$2^k$ machines (say $M_i,i=0,1,2,\ldots,2^{k}-1$)  separately perform the algorithm of quantum counting ${\rm \textbf{Count}}(H^{\otimes (n-k)}, f_i, \lceil\sqrt{2^{n-k}}\rceil)$ to determine a range $R_i=\{b_1^{(i)}, b_2^{(i)},\ldots,b_{r_i}^{(i)}\}$ of $a_i=|\{x\in\{0,1\}^{n-k}|f_i(x)=1\}|$, with $|r_i|\leq 2t_a+1$. If $a_i=0$, then $M_i$ stops, otherwise the rest $M_i$  continue next step in parallel.

Take $j=1$.

$H^{\otimes (n-k)}|0\rangle^{\otimes (n-k)}\rightarrow\frac{1}{\sqrt{2^{n-k}}}\sum\limits_{x\in\{0 , 1\}^{n-k}}|x\rangle$.
	
	 $G_i$ is performed with $\lfloor\frac{\pi}{4}\sqrt{\frac{2^{n-k}}{b_j}} \rfloor$ times, which $G_i=-H^{\otimes (n-k)}Z_0H^{\otimes (n-k)}Z_{f_i}$, $Z_{f_i}|x\rangle=(-1)^{f_i(x)}|x\rangle$, $Z_0|x\rangle=\left\{
\begin{array}{rcl}
-|x\rangle, & &x=0^{n-k} ;\\
|x\rangle, & & x\neq 0^{n-k}.
\end{array} \right.$
	
	Measure the resulting state, and check whether it is a solution of $f_i$. If it is, then end the algorithm, otherwise take $j=j+1$ and return to Step 4.
	
\end{algorithm}

In order to explain the above algorithm, we use the following theorem.

\begin{theorem} Given $n>k\geq 1$, let Boolean function $f:\{0,1\}^n \rightarrow\{0,1\}$ with $1\leq a=|x\in\{0,1\}^{n}|f(x)=1\}|$. Then for any $a_i=|\{x\in\{0,1\}^{n-k}|f_i(x)=1\}|$ with $a_i\geq 1$, Algorithm  \ref{DGAP} can find out the string $x\in\{0,1\}^n$ such that $f(x)=1$ with query times
$\sum_{i=1}^{r_i}  \lfloor \frac{\pi}{4}\sqrt{\frac{2^{n-k}}{b_i}} \rfloor+\lceil\sqrt{2^{n-k}}\rceil+2t_a+1$ for some $1\leq b_i\leq a$ and $r_i\leq 2t_a+1$, and the success probability is at least
\begin{equation}
\frac{8}{\pi^2}P_{G_i},
\end{equation}
where $P_{G_i}=\sin^2{((2\lfloor\frac{ \pi\sqrt{2^{(n-k)}}}{4\sqrt{a_i}} \rfloor+1)\sin^{-1}(\frac{a_i}{\sqrt{2^{(n-k)}}}))}$ is the success probability of Grover's algorithm for computing  $f_i$.

\end{theorem}

\begin{proof}



For any $a_i=|\{x\in\{0,1\}^{n-k}|f_i(x)=1\}|$ with $a_i\geq 1$, the algorithm of quantum counting takes  $\sqrt{2^{(n-k)}}$ queries to determine a range $R_i=\{b_1^{(i)}, b_2^{(i)},\ldots,b_{r_i}^{(i)}\}$ of $a_i=|\{x\in\{0,1\}^{n-k}|f_i(x)=1\}|$, with $|r_i|\leq 2t_a+1$. The success probability is at least $\frac{8}{\pi^2}$. Then we use Grover's algorithm to search for a solution of $f_i$ according to the possible values of goals in  $R_i$ in sequence, and after each measuring, we check whether  the result is the solution of $f_i$. So, it takes at most  $\sum_{i=1}^{r_i}  \lfloor \frac{\pi}{4}\sqrt{\frac{2^{n-k}}{b_i}} \rfloor+\lceil\sqrt{2^{n-k}}\rceil+2t_a+1$ queries for some $1\leq b_i\leq a$ and $r_i\leq 2t_a+1$. Since $a_i\in R_i$, the success probability of Grover's algorithm for computing  $f_i$ is $P_{G_i}$. Consequently, the success probability is at least
$\frac{8}{\pi^2}P_{G_i}$.

\end{proof}

\begin{remark}
Finally, we consider a special case, that is $|x\in\{0,1\}^{n}|f(x)=1\}|=a=1$. In this case, $a_i$ is $0$ or $1$ where $a_i=|\{x\in\{0,1\}^{n-k}|f_i(x)=1\}|$. So, we can use Grover's algorithm directly by just taking $a_i=1$ instead of using the algorithm of quantum counting, and then check whether the measuring result is the solution of $f_i$. Therefore, the number of queries in  the above distributed Grover's algorithm in parallel  is 
\begin{equation}
\lfloor \frac{\pi}{4}\sqrt{2^{n-k}} \rfloor.
\end{equation}

\end{remark}

\section{ Realization of Oracles with Quantum Circuits}
Constructing quantum circuit to realize an oracle is important in Grover's algorithm, since it is a key step from practical problems to applications.
As above, more oracles are required in the distributed Grover's algorithms, so how to construct quantum circuits to realize  oracles physically is useful in practice. In this section, the goal is to design an algorithm for realizing oracles and provide a corresponding quantum circuit.

More specifically,
given any Boolean function $f:\{0,1\}^n\rightarrow \{0,1\}$, its oracle $Z_f$ is defined as $Z_f|x\rangle=(-1)^{f(x)}|x\rangle$, and  the unitary operator $G$ is constructed from $Z_f$ in Algorithm \ref{Grover's Algorithm}.
 Recently, there was a method to  construct a quantum circuit for realizing oracles related to Boolean functions in the form of disjunctive normal form (DNF) proposed in \cite{ACR21}. However, their method relies on the truth table of  functions to be computed, which means that it is difficult to apply in practice, since the truth table of a Boolean function is likely not known in general. Therefore, in this section, we give a method to construct the oracle for any given Boolean function which is a conjunctive normal form (CNF), without knowing its truth table. As is known, 3SAT problem is a special case of conjunctive normal forms.

 A CNF that relies on Boolean variables $x_1,x_2,\cdots,x_n$ is of the form 
\begin{equation} \label{clause}
C_1\land C_2\land \cdots \land C_m,
\end{equation} 
where  $C_i$ (called clause) is in the form $v_1\lor v_2\lor \cdots \lor v_{k_i}$ ($i=1,2,\cdots, m$) and  $v_{j}$ (called literal) is a variable or  negation of a variable in $\{x_1,x_2,\cdots,x_n\}$ ($j=1,2,\cdots, k_i$). For example, $(x_1\lor \neg x_2)\land (\neg x_1 \lor x_4\lor \neg x_3)$ is a CNF with $2$ clauses that relies on variables $x_1,x_2,x_3,x_4$, where $k_1=2$ and $k_2=3$. In the interest of simplification, we consider functions that are 3CNF, where 3CNF is a CNF with $k_i\leq 3,i=1,2,\cdots,m$. 3CNF is vital, since the famous 3SAT problem (the first proven NPC problem) is to determine whether any given 3CNF has a solution. Of course, our method can be used similarly for realizing the oracles corresponding to general conjunctive normal forms.

How can we construct a quantum circuit for realizing the oracle related to a 3CNF? The idea is as follows. We consider such a circuit. We use $n$ qubits to represent $n$ Boolean variables, and use $\lceil \log (m+1)  \rceil$ ancillary qubits to store the number of clauses that are false. Thus, we can know the value of the function under the current assignment of variables is true if and only if the ancillary qubits is in state $|0 \rangle^{\otimes\lceil \log (m+1)  \rceil}$. In addition, we flip the phase if the ancillary qubits is $|0 \rangle^{\otimes\lceil \log (m+1)  \rceil}$, and finally restore the ancillary qubits. This idea can be formalized  by means of Algorithm \ref{Compute 3CNF}, and the corresponding circuit of Algorithm \ref{Compute 3CNF} is shown in Figure \ref{fig:Compute3CNF}.

For describing Algorithm \ref{Compute 3CNF} clearly, we first define two unitary operators $U_k$ and $U_k'$ depending on clause $C_k$ in a CNF as above Eq. (\ref{clause}) ($k=1,2,\ldots,m$), where $U_k$ and $U_k'$ act on the Hilbert space spanned by $\{|i_1,i_2,\ldots,i_n\rangle |j_1,j_2,\ldots,j_M\rangle: i_1,i_2,\ldots,i_n,j_1,j_2,\ldots, j_M\in \{0,1\} ,                    M= \lceil \log (m+1)  \rceil      \}$.

\begin{definition}
For each $k=1,\cdots,m$,
unitary operator $U_k$ is defined as: for any  $|y_1,\cdots,y_n\rangle \in \{|i_1,i_2,\ldots,i_n\rangle: i_1,i_2,\ldots,i_n\in \{0,1\}   \}$ and for any $c\in\{0,1,\cdots,m\}$, 
\begin{align}
&U_k|y_1,\cdots,y_n\rangle |c\rangle_C \\
 =& \left\{
\begin{array}{lcl}
|y_1,\cdots,y_n\rangle |c\rangle_C, & &  C_k \text{ is true under the assignment $y_1,\cdots,y_n$} ;\\
|y_1,\cdots,y_n\rangle |(c+1)\bmod (m+1)\rangle_C, & & C_k \text{ is false under the assignment $y_1,\cdots,y_n$}.
\end{array} \right.
\end{align} 

\end{definition}

\begin{definition}
For each $k=1,\cdots,m$,
unitary operator $U_k'$ is defined as: for any $|y_1,\cdots,y_n\rangle \in \{|i_1,i_2,\ldots,i_n\rangle: i_1,i_2,\ldots,i_n\in \{0,1\}   \}$ and for any $c\in\{0,1,\cdots,m\}$, 
\begin{align}
&U_k'|y_1,\cdots,y_n\rangle |c\rangle_C \\
 =& \left\{
\begin{array}{lcl}
|y_1,\cdots,y_n\rangle |c\rangle_C, & &  C_k \text{ is true under the assignment $y_1,\cdots,y_n$} ;\\
|y_1,\cdots,y_n\rangle |(c-1)\bmod (m+1)\rangle_C, & & C_k \text{ is false under the assignment $y_1,\cdots,y_n$}.
\end{array} \right.
\end{align} 

\end{definition}


\begin{algorithm}[H]
	\caption{Compute 3CNF }\label{Compute 3CNF}
	\LinesNumbered 
	\KwIn{A Boolean function $f$ that is a 3CNF with representation $C_1\land C_2\land \cdots \land C_m$.\\ An assignment  (denoted as $y_1,\cdots,y_n$) of its variables $x_1,\cdots,x_n$.}
	\KwOut{$(-1)^{f(y_1,\cdots,y_n)}|y_1\cdots y_n\rangle$}
    Initialize state $|y_1\cdots y_n\rangle|0\rangle_C$, where register $C$ is ${\lceil \log (m+1) \rceil}$-qubit.
	
	 Apply $U_1,\cdots,U_m$ in sequence. \quad\quad\qquad\qquad\qquad\qquad\qquad\qquad\qquad\qquad\qquad\qquad
	
	 Apply $Z_0$  to register $C$, where $Z_0|x\rangle_C=\left\{
\begin{array}{rcl}
-|x\rangle_C, & &x=0 ;\\
|x\rangle_C, & & x\neq 0.
\end{array} \right.$

	Apply $U_m',\cdots,U_1'$ in sequence.
\quad\quad\qquad\qquad\qquad\qquad\qquad\qquad\qquad\qquad\qquad\qquad
	
\end{algorithm}

\begin{figure}[H]
    \centering
    \includegraphics[scale=0.5]{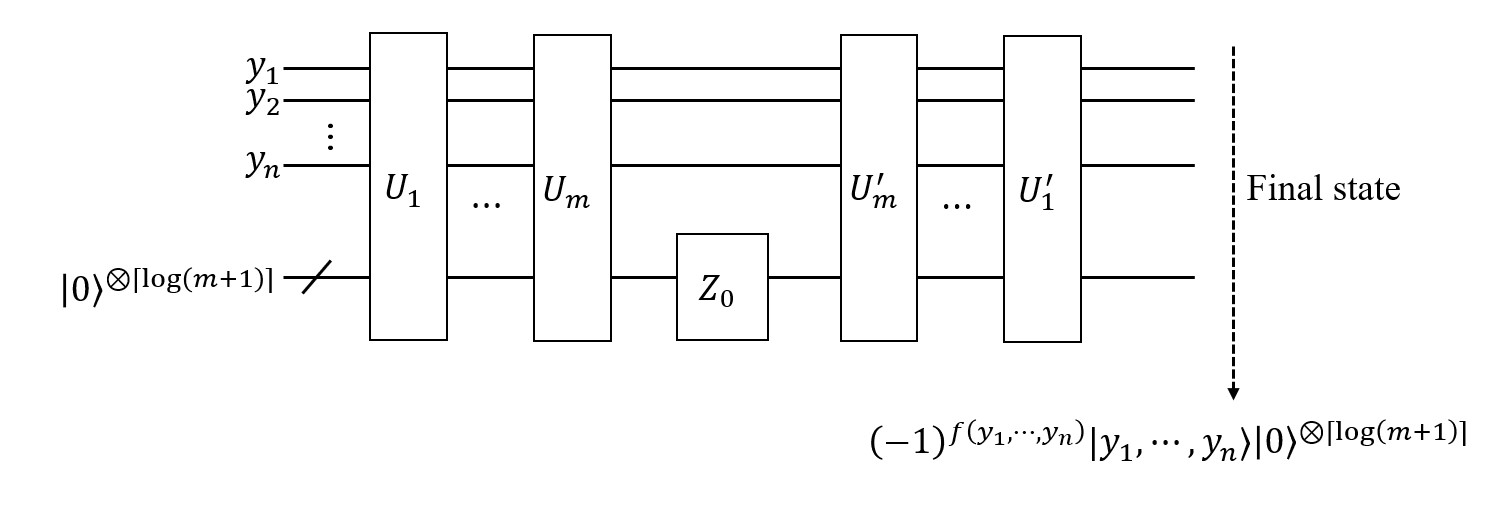}
    \caption{Circuit of Algorithm \ref{Compute 3CNF}}\label{fig:Compute3CNF}
\end{figure}

Next, we give the proof of  correctness of Algorithm \ref{Compute 3CNF}.

\begin{theorem}\label{th_oracle}
The final state of Algorithm \ref{Compute 3CNF} is $(-1)^{f(y_1,\cdots,y_n)}|y_1,\cdots,y_n\rangle|0\rangle^{\otimes\lceil \log (m+1)  \rceil}$.
\end{theorem}
\begin{proof}
Denote $y=(y_1,\cdots,y_n)$, and let $\alpha_1(y),\cdots,\alpha_m(y)$ be defined  as: $$\alpha_k(y)=\left\{
\begin{array}{rcl}
1, & &C_k \text{ is false under the assignment $y_1,\cdots,y_n$} ;\\
0, & &C_k \text{ is true under the assignment $y_1,\cdots,y_n$} ,
\end{array} \right.$$
for $k=1,\cdots,m$.
According to the definition of each $U_k$, we have
\begin{equation}
U_k|y_1\cdots y_n\rangle|c\rangle_C=|y_1\cdots y_n\rangle|c+\alpha_k(y)\rangle_C,
\end{equation} 
for $k=1,\cdots,m$ and $c=0,1,\cdots,m-1$. Therefore we have 
\begin{equation}
U_k\cdots U_1|y_1\cdots y_n\rangle|0\rangle_C=|y_1\cdots y_n\rangle|\sum_{j=1}^k\alpha_j(y)\rangle_C,
\end{equation}
for $k=1,\cdots,m$. Hence, after step 2, the quantum state is 
\begin{equation}
U_m\cdots U_1|y_1\cdots y_n\rangle|0\rangle_C=|y_1\cdots y_n\rangle|\sum_{j=1}^m\alpha_j(y)\rangle_C,
\end{equation}
where $\sum_{j=1}^m\alpha_j(y)\leq m$. Since $f(y_1,\cdots,y_n)=1$  if and only if $\sum_{j=1}^m\alpha_j(y)=0$, we in step 3 obtain 
\begin{equation}
Z_0|\sum_{j=1}^m\alpha_j(y)\rangle_C=(-1)^{f(y_1,\cdots,y_n)}|\sum_{j=1}^m\alpha_j(y)\rangle_C.
\end{equation} 
Similarly, in step 4, we can see that 
\begin{equation}
U_k'\cdots U_m'(-1)^{f(y_1,\cdots,y_n)}|y_1\cdots y_n\rangle|\sum_{j=1}^m\alpha_j(y)\rangle_C=(-1)^{f(y_1,\cdots,y_n)}|y_1\cdots y_n\rangle|\sum_{j=1}^{k-1}\alpha_j(y)\rangle_C,
\end{equation} 
for $k=2,\cdots,m$. When $k=1$, we get  
\begin{equation}
U_1'\cdots U_m'(-1)^{f(y_1,\cdots,y_n)}|y_1\cdots y_n\rangle|\sum_{j=1}^m\alpha_j(y)\rangle_C=(-1)^{f(y_1,\cdots,y_n)}|y_1\cdots y_n\rangle | 0\rangle_C.
\end{equation}  Consequently, we have completed the proof.
\end{proof}

\begin{remark}
We use ${\lceil \log (m+1)  \rceil}$ ancillary qubits to implement the transformation from $|y_1,\cdots,y_n\rangle$ to $(-1)^{f(y_1,\cdots,y_n)}|y_1,\cdots,y_n\rangle$. Moreover, the construction of operators $U_k,U_k'$ $(k=1,\cdots,m)$ only depends on $C_k$. For example, suppose $C_k=(x_1\lor \neg x_3 \lor \neg x_4)$ for some $k$, since $C_k$ is false if and only if $x_1=0,x_3=1$ and $x_4=1$, we know $U_k$ and $U_k'$ can be construted as Figure \ref{Ui}, where the operators $ADD$ and $SUB$ act as follows: 
\begin{equation}ADD|c\rangle_C=|(c+1)\bmod (m+1)\rangle_C,\end{equation}
\begin{equation}SUB|c\rangle_C=|(c-1)\bmod (m+1)\rangle_C\end{equation}
 for any $c\in\{0,1,\cdots,m\}$. The operators $ADD$ and $SUB$ can be constructed by means of using $O(\lceil \log (m+1)  \rceil)$ elementary gates with $O(\lceil \log (m+1)  \rceil) $ ancillary qubits \cite{YL04}. According to the similar method,  the construction of each $U_k$ and $U_k'$ can be generalized to general case of $C_k$ with number of variable more than 3.

\end{remark}

\begin{figure}[H]
    \centering
    \includegraphics[scale=0.5]{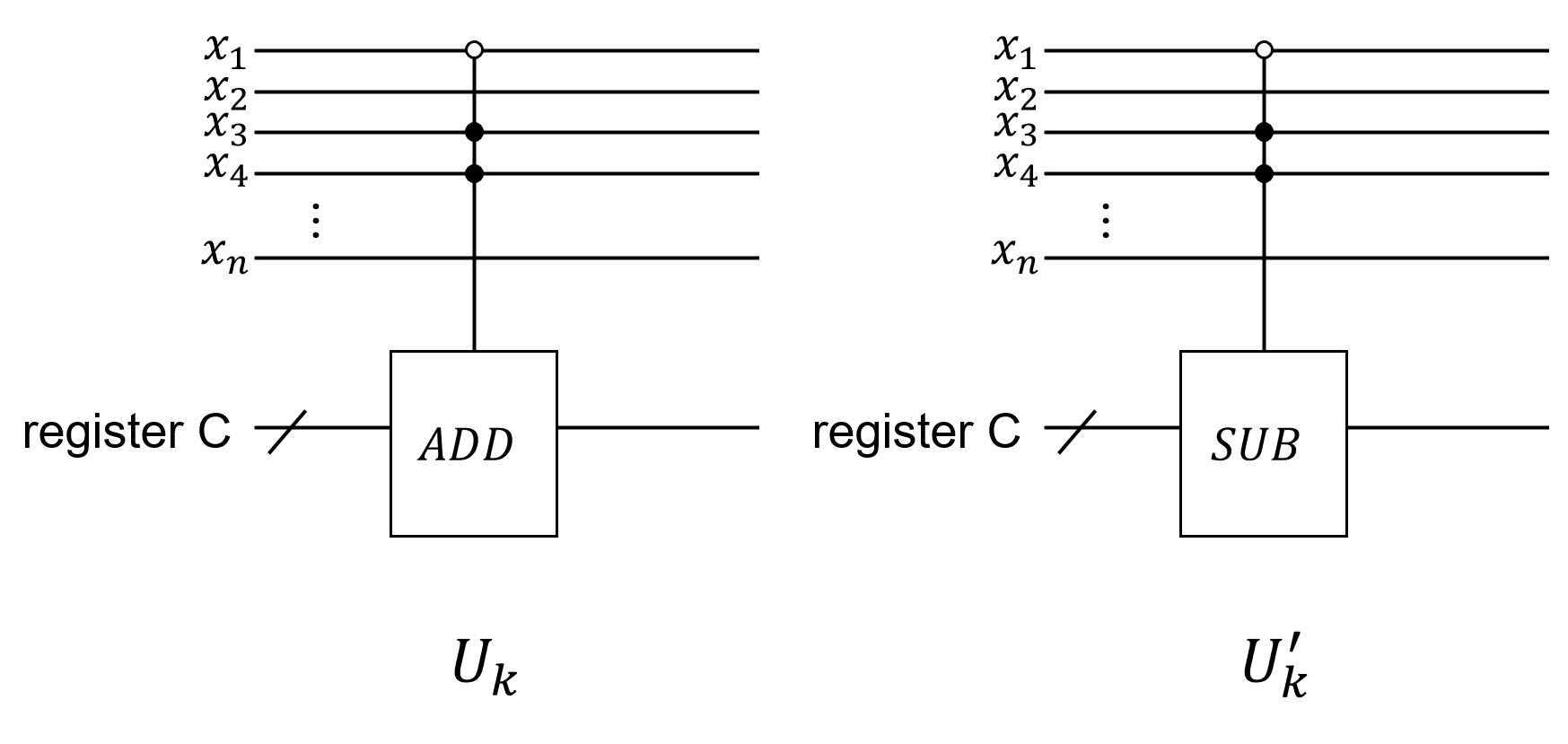}
    \caption{Constructions for $U_k$ and $U_k'$.}\label{Ui}
\end{figure}


Due to  Algorithm \ref{Compute 3CNF}, we have the following theorem.

\begin{theorem}\label{th_oracle_complexity} Let Boolean function $f: \{0,1\}^n\rightarrow \{0,1\}$ be a 3CNF with $m$ clauses. Then the oracle  $Z_f|x\rangle=(-1)^{f(x)}|x\rangle$ can be realized by using $O(m  \log m)$ elementary gates with $O(\log m) $ ancillary qubits.
\end{theorem}
\begin{proof}  By means of Algorithm \ref{Compute 3CNF}, 
the oracle  $Z_f|x\rangle=(-1)^{f(x)}|x\rangle$ can be realized by operators $U_1,U_2,\cdots,U_m,Z_0$,$U_1',U_2',\cdots,U_m'$.

 The operator $ADD$ or $SUB$ can be constructed in terms of using $O(\lceil \log (m+1)  \rceil)$ elementary gates with $O(\lceil \log (m+1)  \rceil) $ ancillary qubits \cite{YL04}. 
 Hence, each $U_k$ or $U_k'$ can be constructed by using $O(\lceil \log (m+1)  \rceil)$ elementary gates with $O(\lceil \log (m+1)  \rceil) $ ancillary qubits. $Z_0$ is a muti-qubit controlled gate and can be constructed by using $O(\lceil \log (m+1)  \rceil)$ elementary gates with fixed number of ancillary qubits \cite{ACR21,BBC95}. So, the theorem holds.
\end{proof}

\begin{remark} \label{DOC}
In the distributed Grover's algorithm we have designed, the oracles related to all Boolean subfunction $ f_i(x)=f(xy_i)$ ($i=1,2,\ldots, 2^t-1$) are needed,  where $f:\{0,1\}^n\rightarrow \{0,1\}$ is given, $y_i\in\{0,1\}^t$ is the binary representation of $i$ and $ x\in\{0,1\}^{n-t}$. If $f$ is a 3CNF, then $f_i$ is also a 3CNF and usually can be simplified with classical method  based on the known variables $y_i$ to the 3CNF of $f$. Since $f_i$ is a 3CNF, by Theorem \ref{th_oracle_complexity}, we can realize its oracle with $O(m_i \log m_i)$ elementary gates and $O(\log m_i)$ ancillary qubits, where $m_i$ is the number of clauses in the 3CNF of $f_i$. Therefore, for any positive integer $t$, the time complexity of realizing the oracles for $f_0,f_1,\cdots,f_{2^t-1}$ 
in parallel is $O(\max_{i\in\{0,1,\cdots,2^t-1\}} m_i\log m_i)$.

\end{remark}

\section{Concluding remarks}

It is still difficult to design large-scale universal quantum computers nowadays due to significant depth of quantum circuits and noise impact. Therefore, quantum models and algorithms with smaller number of input bits and shallower circuits have better physical realizability.  Distributed quantum algorithms usually require smaller size of input bits and better time complexity as well as shallower depth of quantum circuits, so it is intriguing to design distributed quantum algorithms for important quantum algorithms such as Shor's algorithm and Grover's algorithm.

In this paper, we have given a distributed Grover's algorithm. We can divide a function to be computed into $2^k$ subfunctions and then compute one of useful subfunctions to obtain the solution of the original function. For dealing with these subfunctions, we have considered two methods, i.e., serial method and parallel method. The success probabilities of two methods are same and approximate to that of Grover's algorithm. As for query times (depth of quantum circuits) in the worst case, parallel method needs at most  $\sum_{i=1}^{r_i}  \lfloor \frac{\pi}{4}\sqrt{\frac{2^{n-k}}{b_i}} \rfloor+\lceil\sqrt{2^{n-k}}\rceil+2t_a+1$ for some $1\leq b_i\leq a$, where $r_i\leq 2t_a+1$ and $t_a=\lceil 2\pi\sqrt{a}+11\rceil$, versus $\lfloor \frac{\pi}{4}\sqrt{\frac{2^{n}}{a}} \rfloor$ of Grover's algorithm, where $n$ is the number of input bits of problem, but the number of input bits for subfunctions is $n-k$. 


In particular, if $|x\in\{0,1\}^{n}|f(x)=1\}|=a=1$, then our distributed Grover's algorithm in parallel only needs $\lfloor \frac{\pi}{4}\sqrt{2^{n-k}} \rfloor$ queries, versus $\lfloor \frac{\pi}{4}\sqrt{2^{n}} \rfloor$ queries of Grover's algorithm.

For any Boolean function $f$ with CNF having $m$ clauses, we have proposed an efficient algorithm of constructing quantum circuit for  realizing the oracle $Z_f|x\rangle=(-1)^{f(x)}|x\rangle$ corresponding to $f$. The algorithm's time complexity is $O(m  \log m)$. According to Remark \ref{DOC},  the time complexity of realizing the oracles in our distributed Grover's algorithm 
in parallel is at most $O( m\log m)$, where $m$ is the number of clauses in the Boolean function to be computed.


Other  problems worthy of further consideration   are how to design distributed quantum algorithms for  Deutsch-Jozsa problem, Hidden subgroup problem,  and  decomposition of large number. We would like to consider it in sequent study.    In particular,  we would like to realize these distributed quantum algorithms with ion trap quantum computers  experimentally in near future.

\section*{Acknowledgements}

We thank Dr. Markus Grassl for pointing out a problem of query times in Theorem 1.
 This work is supported in part by the National
Natural Science Foundation of China (Nos.   61876195, 61572532) and the Natural Science
Foundation of Guangdong Province of China (No. 2017B030311011).


\begin{thebibliography}{AB}

\bibitem{ACR21}
J. Avron, O. Casper, and I. Rozen, Quantum advantage and noise reduction in distribute quantum computing, Physical Review A 104 (2021) 052404.


 \bibitem{Ben80}P. Benioff, The computer as a physical system: a microscopic quantum mechanical Hamiltonian model of computers as
represented by Turing machines,  Journal  of\ Statistic\ Physics
22 (1980) 563-591.\

\bibitem{BBC95}
A. Barenco, C. H. Bennett, R. Cleve, et al. Elementary gates for quantum computation. Physical Review A 52 (5) (1995) 3457.

\bibitem{BBG13} R. Beals, S. Brierley, O. Gray, et al., 
Efficient distributed quantum computing, Proceedings of the Royal Society A 469 (2013) 20120686.


\bibitem{BHMT02} G. Brassard, P. Hoyer, M. Mosca, and A. Tapp, Quantum amplitude amplification and estimation, Contemp. Math. 305 (2002) 53-74.


\bibitem{BV97} E. Bernstein, U. Vazirani,  Quantum complexity theory, SIAM Journal on Computing
 26 (5) (1997) 1411-1473.\

\bibitem{CEMM98}
R.~Cleve, A.~Eckert, C.~Macchiavello, and M.~Mosca,
 Quantum algorithms revisited,
Proceedings of the Royal Society of London 454A (1998) 339--354.



 \bibitem{CQ18} G. Cai, D.W. Qiu, Optimal separation in exact query complexities for Simon problem, Journal of Computer and System Sciences 97 (2018) 83-93.

\bibitem{Deu85} D. Deutsh, Quantum theory, the Church-Turing principle and the universal quantum computer,  Proceedings of the Royal Society of London Series A
 400 (1985) 97-117.\
\bibitem{Deu89} D. Deutsh,  Quantum computational networks,  Proceedings of the Royal Society of London Series A
 400 (1985) 73-90.


  \bibitem{DJ92}
D.~Deutsch and  R.~Jozsa,
 Rapid solution of problems by quantum computation,
Proceedings of the Royal Society of London 439A (1992) 553--558.


\bibitem{Fey82} R.P. Feynman,  Simulating physics with computers,  International \ Journal of Theoretical\ Physics  21 (1982)
467-488.\


\bibitem{Gro96}L.K. Grover, A fast quantum mechanical algorithm for database search,
in: Proceedings of the 28th  Annual ACM Symposium on Theory of
Computing, Philadelphia, Pennsylvania, USA, 1996, pp. 212-219.\

\bibitem{KLM07} P. R. Kaye, R. Laflamme, and M. Mosca, \emph{An Introduction to Quantum Computing}, Oxford University Press, New York,  2007.


\bibitem{Long01} G. L. Long, Grover algorithm with zero theoretical failure rat, Physical  Review A 64 (2001) 022307.

\bibitem{LQ17} K. Li, D.W. Qiu, et al., Application of distributed semi-quantum computing model in phase estimation, Information Processing Letters 120 (2017) 23-29.



\bibitem{MW18}A. Molina, J. Watrous, Revisiting the simulation of quantum Turing
machines by quantum circuits, Proceedings of the Royal Society of London Series A 475 (2019) 20180767.




\bibitem{NC00} M. Nielsen, I. Chuang, Quantum Computation and Quantum Information, Cambridge University Press, Cambridge, 2000.

\bibitem{PQ19} M. Pan, D.W. Qiu,  Operator coherence dynamics in Grover's quantum search algorithm, Physical Review A 100(1) (2019) 1-10.




\bibitem{QZ18} D.W. Qiu, S. Zheng, Generalized Deutsch-Jozsa problem and the optimal quantum algorithm, Physical  Review A 97 (2018) 062331.

\bibitem{QZ20} D.W. Qiu, S. Zheng, Revisiting Deutsch-Jozsa Algorithm, Information and Computation, 2020,275: 1-12.

\bibitem{Sho97} P.W. Shor, Polynomial-time algorithms for prime
factorization and discrete logarithms on a quantum computer, SIAM
Journal on Computing 26 (5) (1997) 1484-1509.




\bibitem{TXQ22} J. Tan, L. Xiao, D.W. Qiu, L. Luo, and P. Mateus, Distributed quantum algorithm for Simon's problem,
 Physical  Review A  106 (2022) 032417.





\bibitem{Yao93}A.C. Yao, Quantum circuit complexity,  in: Proceedings of the 34th IEEE Symposium on Foundations of Computer science, 1993,  pp.
352-361.\


\bibitem{YL04}
A. Yimsiriwattana, S.J. Lomonaco, Distributed quantum computing: A distributed Shor algorithm, Quantum Information and Computation II  5436 (2004) 360-372.































\end{thebibliography}
\end{document}